\documentclass[11pt,letterpaper]{article}
\usepackage[T1]{fontenc}
\usepackage[utf8]{inputenc}
\usepackage{lmodern}
\usepackage[left=1.05in,right=1.05in,top=1in,bottom=1.1in,includefoot]{geometry}

\usepackage{amsmath,amssymb,amsthm,booktabs,tabularx,array,graphicx,microtype}
\usepackage{tikz,flafter,float,placeins}
\usetikzlibrary{arrows.meta,positioning}
\usepackage[round,authoryear]{natbib}
\usepackage{hyperref,url}
\makeatletter
\g@addto@macro\UrlBreaks{\do\a\do\b\do\c\do\d\do\e\do\f\do\g\do\h\do\i\do\j\do\k\do\l\do\m\do\n\do\o\do\p\do\q\do\r\do\s\do\t\do\u\do\v\do\w\do\x\do\y\do\z}
\makeatother
\microtypesetup{protrusion=false}
\newcommand{\wtb}{\textsc{WTB}}
\newcommand{\cont}{\textsc{Continuity}}
\newcommand{\mas}{\textsc{MasDrift}}
\newcommand{\code}[1]{\texttt{#1}}
\makeatletter
\let\tableinput\@@input
\makeatother
\newtheorem{proposition}{Proposition}
\title{From Evidence to Effect:\\Authority Semantics and Runtime Infrastructure\\for Stateful Agents}
\author{\normalsize\begin{tabular}{c}
Yang Li\textsuperscript{1}\quad Zongsi Xu\textsuperscript{1}\quad Sergey Volkov\textsuperscript{1}\quad hai liu\textsuperscript{2}\\[3pt]
Tuo Zhou\textsuperscript{1}\quad Xiyu Chen\textsuperscript{1}\quad Di Wan\quad Dian Shao\textsuperscript{4}\\[3pt]
Ye Luo\textsuperscript{1,*}\quad Hao Sun\textsuperscript{5,*}
\end{tabular}\\[8pt]
\small\begin{tabular}{c}
\textsuperscript{1}University of Hong Kong (hku.hk)\\[2pt]
\textsuperscript{2}Jiangxi Science and Technology Normal University (jxstnu.edu.cn)\\[2pt]
\textsuperscript{4}The Hong Kong University of Science and Technology (ust.hk)\\[2pt]
\textsuperscript{5}Shenzhen University (szu.edu.cn)
\end{tabular}\\[6pt]
\small\textsuperscript{*}Corresponding authors: Ye Luo and Hao Sun.}
\date{}
\hypersetup{hidelinks,pdftitle={From Evidence to Effect: Authority Semantics and Runtime Infrastructure for Stateful Agents},pdfauthor={Yang Li; Zongsi Xu; Sergey Volkov; hai liu; Tuo Zhou; Xiyu Chen; Di Wan; Dian Shao; Ye Luo; Hao Sun}}
\begin{document}
	\maketitle

\begin{abstract}
	Stateful agents reuse artifacts after producing executions and permissions change. We formalize authority-sufficient observations and durable effects bound to execution and material identities. WTB implements this interface through runtime adapters, shared evidence, and transactional publication/recovery. Six study families separate the mechanism from its integration. Raw and typed evidence both solve 32/32 authority cases, with model-dependent planning effects. Fixed-intent enforcement blocks six unsafe proposals and executes 12 eligible authorized intents. Complete controls match WTB's capability. Paid integration yields 176/210 accepted benchmark-source stages, including 19/30 publication stages, recovery on 8/8 primary SWE repositories, and the most complete continuous trajectories on each of three source tasks. The findings connect authority information, effect admission, and infrastructure reuse in stateful agents.
\end{abstract}

\section{Introduction}
\label{sec:intro}
Consider two agents that contribute to a configuration. Their final file can pass every functional test while one contribution comes from a superseded attempt or a writer whose permission has been revoked. A later publisher must resolve that relation across the workspace, execution registry, and authority service. If the remote write succeeds but its response is lost, recovery must preserve the same relation while determining whether an effect already exists. Artifact correctness and execution authority are complementary obligations.

Our central claim is that \emph{stateful agent effects need an explicit, execution-bound authority interface}. Its meaning depends on both historical evidence and current state. We separate three questions: which facts make authority observable, how their representation affects planning, and which transitions preserve authority when an effect commits. This separation explains why accurate planning and safe execution require different interventions.

\paragraph{Theoretical contribution.}
We connect observation sufficiency, execution/material identity, and commit-time authority in one stateful-agent contract. This yields a testable separation between adding authority facts, changing their representation, and admitting a fixed intent. The novelty is the execution-to-effect account and its intervention design: standard observation factorization and complete mediation supply its foundations \citep{saltzer,continuity}.

\paragraph{Engineering contribution.}
\wtb{} connects runtime evidence to publication authority through a common attempt-binding interface, shared receipt store, and transactional publication/recovery. DSH/Cordis and LangGraph adapters, plus a Ray capture prototype, connect native identities to that core (Figure~\ref{fig:architecture}). Complete controls tie on capability; WTB improves observed fixed-budget completion in acquisition, recovery, and handoff. The coding studies locate the practical value in reuse of the supplied authority interfaces.

\paragraph{Experimental scope.}
We report six study families and 15 protocol/cohort subexperiments across four benchmark families and 16 executed task/repository workloads. The record includes 1,088 controlled authority/runtime cells, 96 planning calls, 32 matched-intent replays, 392 equal-contract checks, 24 live interleavings, 1,120 independent coding slots, 120 prospective SWE stage slots, and 117 continuous trajectories with 468 stage outcomes. Table~\ref{tab:inventory} preserves their units. All six families contribute primary-text evidence; the appendices provide the full source-level analyses.

\section{Authority semantics}
\label{sec:method}
\paragraph{Authority-sufficient observations.}
Let $S=(W,E,A,L)$ contain workspace/material state $W$, execution evidence $E$, scoped authority $A$, and the publication ledger $L$. For an effect request $a$, let $Y(S,a)\in\{0,1\}$ denote admissibility. An observation $O$ is \emph{authority-sufficient} on a state domain if some decision rule $d$ satisfies $d(O(S),a)=Y(S,a)$ throughout that domain.

\begin{proposition}[Observation sufficiency]
	An observation is authority-sufficient exactly when $Y(\cdot,a)$ is constant on each observation-equivalence class, for every request $a$.
	\label{prop:sufficiency}
\end{proposition}
The proof is the factorization of $Y$ through $O$ (Appendix~\ref{app:aliasing}). It supplies a construction criterion: hold visible material and history fixed, then vary a hidden authority relation. For example, exchanging two writers' scope assignments preserves their final bytes but reverses admissibility. A planner using the unchanged observation has expected accuracy $1/2$ on a balanced pair requiring opposite decisions.

\paragraph{Three interventions.}
Adding authoritative evidence can split an aliased class. Lossless raw and typed encodings of the same admitted facts preserve its information partition, while an LLM can respond differently to their layouts. Finally, a valid decision at time $t$ can become stale before commit if an attempt, permission, material version, or approval state changes. These observations motivate three controlled contrasts:
\begin{equation}
	\underbrace{\mathrm{C0}\to\mathrm{R0}}_{\text{authority facts}},\qquad
	\underbrace{\mathrm{R0}\to\mathrm{W0}}_{\text{equal-information representation}},\qquad
	\underbrace{\mathrm{W0}\to\mathrm{W1}}_{\text{fixed-intent effect admission}}.
	\label{eq:contrasts}
\end{equation}
The first two change planner observations; the third fixes its output and changes effect admission. Observation sufficiency characterizes information, while commit-time admission accounts for mutable state.

\section{WTB: binding evidence to durable effects}
\label{sec:infrastructure}
\paragraph{Execution-bound authority.}
A publication binding $b=(C,m,v,q)$ contains contributor records $C$, material $m$ with admitted version $v$, and approval $q$. Each contributor record identifies a native execution $x$, admitted attempt $k$, actor $u$, scope $s$, terminal receipt $r$, and scoped authority generation $g$. Define $\mathrm{Fresh}(C,S)$ to require every contributor's current attempt, completed matching receipt, and current permission at the recorded generation. Fresh publication requires
\begin{equation}
	\begin{split}
		\mathrm{Admit}(b,S)={}&\mathrm{Fresh}(C,S)\land\mathrm{MaterialChainMatches}(m,v,C,S)\\
		&\land\mathrm{ApprovalMatches}(q,C,m,v)\land\mathrm{Unconsumed}(q,S).
	\end{split}
	\label{eq:contract}
\end{equation}
The material predicate binds exact identity and contributor-chain continuity. Equal bytes can belong to different admitted versions; an approval for v1 therefore requires renewal for v2. A scoped generation distinguishes uninterrupted permission from revoke-and-regrant. Both choices become directly observable in our experiments.

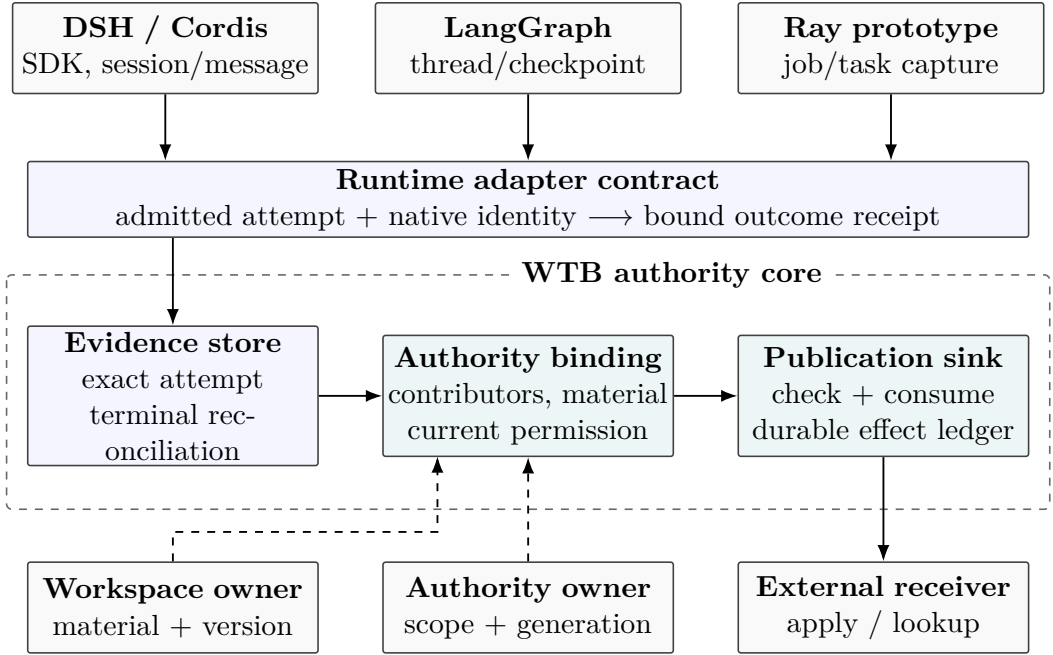
\begin{figure}[t]
	\centering
	\begin{tikzpicture}[x=1mm,y=1mm,font=\normalsize,>=Latex,
  box/.style={draw=black!75,line width=.6pt,rounded corners=1pt,align=center,inner sep=2pt},
  flow/.style={->,line width=.7pt},query/.style={->,dashed,line width=.7pt}]
\node[box,text width=39mm,minimum height=12mm,fill=gray!5] (dsh) at (21,0) {\textbf{DSH / Cordis}\\SDK, session/message};
\node[box,text width=39mm,minimum height=12mm,fill=gray!5] (lg) at (69,0) {\textbf{LangGraph}\\thread/checkpoint};
\node[box,text width=39mm,minimum height=12mm,fill=gray!5] (ray) at (117,0) {\textbf{Ray prototype}\\job/task capture};
\node[box,text width=131mm,minimum height=10mm,fill=blue!4] (port) at (69,-20) {\textbf{Runtime adapter contract}\\admitted attempt + native identity $\longrightarrow$ bound outcome receipt};
\draw[flow] (dsh.south)--(dsh.south |- port.north);
\draw[flow] (lg.south)--(port.north);
\draw[flow] (ray.south)--(ray.south |- port.north);
\draw[draw=black!60,dashed,line width=.6pt,rounded corners=2pt] (0,-30) rectangle (138,-61);
\node[fill=white,inner sep=2pt] at (88,-30) {\textbf{WTB authority core}};
\node[box,text width=37mm,minimum height=16mm,fill=blue!4] (ev) at (22,-46) {\textbf{Evidence store}\\exact attempt\\terminal reconciliation};
\node[box,text width=37mm,minimum height=16mm,fill=teal!7] (bind) at (69,-46) {\textbf{Authority binding}\\contributors, material\\current permission};
\node[box,text width=37mm,minimum height=16mm,fill=teal!7] (sink) at (116,-46) {\textbf{Publication sink}\\check + consume\\durable effect ledger};
\draw[flow] (port.south -| ev.north)--(ev.north);
\draw[flow] (ev.east)--(bind.west);
\draw[flow] (bind.east)--(sink.west);
\node[box,text width=37mm,minimum height=12mm,fill=gray!5] (ws) at (22,-74) {\textbf{Workspace owner}\\material + version};
\node[box,text width=37mm,minimum height=12mm,fill=gray!5] (policy) at (69,-74) {\textbf{Authority owner}\\scope + generation};
\node[box,text width=37mm,minimum height=12mm,fill=gray!5] (receiver) at (116,-74) {\textbf{External receiver}\\apply / lookup};
\draw[query] (ws.north)--(22,-64)--(57,-64)--(57,-54);
\draw[query] (policy.north)--(bind.south);
\draw[flow] (sink.south)--(receiver.north);
\end{tikzpicture}
	\caption{WTB authority architecture. Solid arrows carry bound evidence or effect requests; dashed arrows supply current material and permission facts. Adapters preserve native identity, while the shared core evaluates authority. The sink records a local effect or reserves external delivery for a cooperating receiver. Runtime-specific paths appear in Figure~\ref{fig:adapters}.}
	\label{fig:architecture}
\end{figure}

\paragraph{Adapter contract and evidence core.}
WTB admits an attempt with an exact session, provider identity, configuration, and input before dispatch. The adapter returns a receipt bound to that envelope. DSH uses its Python SDK/JSON-RPC seam and durable message identity; LangGraph uses a workflow facade and checkpoint identity; the Ray prototype captures a real task behind a driver-owned dispatch fence. The shared collector checks identity and records evidence; publication evaluates permission. Monotone terminal reconciliation tolerates late/duplicate events and holds conflicts for resolution. Appendix~\ref{app:adapters} documents each adapter's recovery capabilities.

\paragraph{Publication and recovery.}
For a local sink, one transaction checks Equation~\ref{eq:contract}, consumes $q$, and records the effect. For an external sink, a transaction checks and consumes approval while reserving an immutable effect key, payload, and authority snapshot. The receiver atomically checks its current authority token and inserts at most one effect per key. If the response is lost, recovery looks up that key: an applied result returns its receipt; confirmed absence permits retry after rechecking current contributor/material authority; an unresolved result remains pending (Figure~\ref{fig:contract}). Reusing a receipt produces no new write.

\begin{proposition}[Effect safety under mediated transitions]
	Under trusted evidence and authority state, serialized local check/consume/commit, and atomic receiver token-check/deduplication, each new local effect satisfies Equation~\ref{eq:contract} in its atomic transition's pre-state. Each external effect has a valid reservation and passes the receiver's authority check at commit. A fixed effect key produces at most one committed effect.
	\label{prop:safety}
\end{proposition}
The proof follows the permitted transitions and unique effect keys (Appendix~\ref{app:protocolproof}). The two commit points give distinct local and receiver guarantees. WTB, the SQL guard, and the \cont{}-bound control share the cooperating receiver in our tests. Trusted adapters, actor assignment, stores, policy state, and mediated sinks define the protocol boundary. Tests exercise these responsibilities across implementations; the propositions describe the reference semantics.

\section{Evaluation design}
\label{sec:design}
\begin{table}[t]
	\caption{Six families answer five research questions. RQ2 combines deterministic responsibility ablation and live interleavings. Counts retain their experimental unit; native assertions within a coding stage remain within that stage.}
	\label{tab:inventory}
	\centering
	\begin{tabularx}{\linewidth}{@{}l c X@{}}
		\toprule
		Family & RQ & Design and unit \\
		\midrule
		Original Mechanism & 1 & 128 authority cells; 96 planning calls; 32 fixed-intent replays; 392 contract checks \\
		Runtime Responsibility Ablation & 2 & 3 runtimes $\times$ 2 backends $\times$ 20 worlds $\times$ 8 implementations = 960 cases \\
		Live Authority Interleaving & 2 & 2 runtimes $\times$ 2 backends $\times$ 3 controls $\times$ 2 interventions = 24 executions \\
		Source-Derived Adoption Cost & 3 & 4 source cohorts $\times$ 10 blocks $\times$ 4 arms $\times$ 7 stages = 1,120 slots \\
		SWE Repository Integration & 4 & Three-project calibration; 10 qualified repositories $\times$ 4 arms $\times$ 3 stages = 120 prospective slots \\
		Continuous Adaptation & 5 & 117 observed trajectories on 3 source tasks; 4 sequential stages = 468 outcomes \\
		\bottomrule
	\end{tabularx}
\end{table}

\paragraph{Sources and task definition.}
TeamBench, BulkPR-Bench, and Terminal-Bench each supply one fixed source task \citep{teambench,bulkpr,terminalbench}. Table~\ref{tab:sources} records functional validation and preliminary enforcement; BulkPR's official hidden-verifier result is unavailable. SWE-rebench supplies ten qualified repositories and three calibration projects with native tests \citep{swerebench}. Downstream coding measures governance integration around these reference artifacts. Repeated blocks measure within-task variation; source revisions appear in Appendix~\ref{app:sources}.

\begin{table}[t]
	\caption{Three fixed benchmark-source artifacts support the downstream coding studies. Preliminary W1 task replays use zero paid model calls and remain outside formal study totals. The fresh-authority reference matches WTB on both unsafe and legal effects.}
	\label{tab:sources}
	\centering
	\begin{tabularx}{\linewidth}{@{}X l r r@{}}
		\toprule
		Source task & Saved functional result & \multicolumn{2}{c}{Preliminary W1 replay}\\
		& & Unsafe: unguarded $\to$ WTB & WTB legal \\
		\midrule
		TeamBench: schema evolution & 12/12 checks & $1/3\to0/3$ & 2/2 \\
		BulkPR: paired attrs PRs & 164 local tests pass & $2/3\to0/3$ & 1/1 \\
		Terminal: batched evaluation & 5 tests pass; reward 1 & $1/3\to0/3$ & 2/2 \\
		\bottomrule
	\end{tabularx}
\end{table}

\paragraph{Baselines and substrates.}
Paid coding compares supplied WTB, Git/Auth, \cont{}, and \mas{} packages on DSH. Application starts differ by study: independent adoption supplies stage-specific predecessors; prospective SWE begins with empty integration code; continuous adaptation supplies qualified E2 implementations. Native libraries stay read-only. Appendix~\ref{app:starts} specifies supplied APIs and remaining work. The estimand is integration above these packages, whose responsibilities differ. Original/enhanced contract checks separately test capability. Runtime controls are WTB, SQL guard, and CONTINUITY-bound adapter; DSH/Cordis, LangGraph, and Ray are substrates, with live tests on the first two.

\paragraph{Models, budgets, and inference.}
Authority cells use DeepSeek V4 Flash/GPT-5.4 Mini; planning uses DeepSeek V4 Flash/GPT-5.4 Nano. Paid adoption, prospective SWE, and continuous stages use DeepSeek V4 Flash with 64 responses, 64 actions, and three acceptance submissions. Effect replay, runtime, and contract studies use deterministic payloads without paid model calls. Cost includes valid response work on accepted and failed attempts above supplied libraries; historical construction and human effort are unpriced. We separate fixed-workload, per-accepted-unit, and joint-success costs. Paired tests use repositories or task--blocks, with six SWE and nine continuous Holm comparisons.

\section{RQ1: What changes with information, representation, and enforcement?}
\label{sec:controlled}
\begin{table}[t]
	\caption{Information and representation have distinct empirical effects. (a) Raw authority receipts R0 and typed bindings W0 tie on controlled semantic success. (b) Their unsafe-proposal direction differs across models; unknown decisions are an additional outcome. The cohorts in (a) and (b) have different protocols and model routes.}
	\label{tab:mechanism}
	\centering
	\begin{tabular}{@{}lrr@{}}
		\toprule
		\multicolumn{3}{l}{(a) Authority matrix: 32 cells per arm}\\
		Arm & Final semantic success & ITT composite \\
		\midrule
		\tableinput tables/authority_main.tex
		\bottomrule
	\end{tabular}
	\medskip
	
	\begin{tabular}{@{}lrrrr@{}}
		\toprule
		\multicolumn{5}{l}{(b) Shared-workspace planning: 32 calls per arm}\\
		& \multicolumn{3}{c}{Unsafe publish proposals / unauthorized cases} & Unknown \\
		Arm & Pooled & DeepSeek & GPT-5.4 Nano & / all calls \\
		\midrule
		\tableinput tables/planning_main.tex
		\bottomrule
	\end{tabular}
\end{table}

\paragraph{Authority facts resolve the controlled ambiguity.}
Proposition~\ref{prop:sufficiency} predicts that added authority facts can resolve an observational ambiguity. In the 128-cell matrix, G0 supplies Git evidence with the candidate unavailable; C0 makes candidate generations addressable without authority facts; R0 adds raw receipts; W0 encodes the same admitted facts as a typed relation. Within 32 task/lineage/model/world strata, R0 and W0 achieve 32/32 final semantic successes, and G0 and C0 achieve 0/32 (Table~\ref{tab:mechanism}a). The informative contrast is C0--R0: 32 favorable discordances and zero unfavorable, exact two-sided $p=4.66\times10^{-10}$. The R0--W0 tie attributes this controlled gain to authority information. Arm-specific first-action scores are defined separately in Appendix~\ref{app:original}.

\paragraph{Model-facing packaging has model-dependent effects.}
The separate 96-call study has 16 authorized and 16 unauthorized situations per arm. Pooled unsafe proposals are 12/16 for GM0, 9/16 for R0, and 6/16 for W0. DeepSeek changes from 8/8 unsafe under R0 to 0/8 under W0, while GPT-5.4 Nano changes from 1/8 to 6/8 (Table~\ref{tab:mechanism}b). DeepSeek supplies all 11 W0 unknowns, including seven of its eight unauthorized situations. Thus equal-information encodings change both proposal safety and intent availability; the model interaction precludes a uniform behavioral advantage for typed evidence.

\paragraph{A matched-intent intervention isolates enforcement.}
The native gate receives the same 32 saved W0 outputs. It produces \textbf{0/16 unauthorized effects}, blocking all six unsafe proposals, and executes \textbf{12/12 eligible authorized intents}. Four other authorized-world rows lack a valid publish intent. The six discordant unsafe proposals give exact paired $p=0.03125$. This intervention isolates effect admission from planning: a supplied intent is checked against authority, while a missing intent remains an availability failure.

\begin{table}[t]
	\caption{Enhanced comparison mechanisms reach capability parity. Each configuration receives 56 checks across 14 scenarios and four tasks. WTB passes 56/56. The seven configurations yield 392 outcomes; enhancement-development effort is outside this experiment.}
	\label{tab:parity}
	\centering
	\begin{tabular}{@{}lrr@{}}
		\toprule
		Comparison architecture & Original contract & Enhanced contract \\
		\midrule
		\tableinput tables/parity_main.tex
		\bottomrule
	\end{tabular}
\end{table}

\paragraph{Complete contracts are implementable across architectures.}
Table~\ref{tab:parity} supplies the omitted execution, material, and current-authority responsibilities to each original comparison. All enhanced systems and WTB pass 56/56. This establishes implementability at multiple architectural placements. The coding studies then evaluate the documented library/scaffold packages: their observed differences combine supplied functionality with the work required to connect it to an application.

\section{RQ2: Which runtime responsibilities preserve authority?}
\label{sec:runtime}
\begin{table}[t]
	\caption{Five responsibility removals expose different failures. The 960 deterministic cases draw on six native captures across 20 worlds and eight implementations. ``Legal-correct'' requires the expected completed state and effect count; a duplicate write can set completion while failing this metric.}
	\label{tab:ablations}
	\centering
	\begin{tabular}{@{}lrrr@{}}
		\toprule
		Implementation & Correct cases & Unsafe effects & Legal-correct \\
		\midrule
		\tableinput tables/runtime_main.tex
		\bottomrule
	\end{tabular}
	\medskip
	
	\begin{tabular}{@{}lrrr@{}}
		\toprule
		Live runtime & Correct executions & Relevant-scope effects & Unrelated-scope effects \\
		\midrule
		DSH/Cordis & 12/12 & 0 in 6 executions & 6 in 6 executions \\
		LangGraph & 12/12 & 0 in 6 executions & 6 in 6 executions \\
		\bottomrule
	\end{tabular}
\end{table}

\paragraph{Deterministic ablations identify responsibility-specific failures.}
The protocol predicts distinct consequences from removing binding, freshness, consumption, and reconciliation. We cross DSH/Cordis, LangGraph, and Ray with SQLite and PostgreSQL, drawing on two native attempts per runtime. Each runtime/backend pair evaluates 20 worlds under eight implementations. All complete controls pass 120/120, preserve 60 legal cases, and produce zero unsafe effects (Table~\ref{tab:ablations}). Removing exact-attempt binding causes 12 unsafe effects; generation removal causes six after revoke-and-regrant; consumption removal causes 12 from replay. Monotone-reconciliation removal creates six unsafe effects and loses 12 legal completions. Unknown-recovery removal loses six recoverable completions without creating unsafe effects. The latter distinction is important: conservative uncertainty handling preserves safety while receiver reconciliation restores availability.

\paragraph{Live interleavings exercise authority freshness during execution.}
A separate 24-execution study changes authority during a real in-flight call, then captures its receipt and attempts publication. Across two runtimes, two backends, and three controls, all 12 relevant changes block publication and all 12 unrelated changes preserve one effect. Each implementation passes 8/8. This tests the gap between an earlier admission and current publication: freshness blocks relevant changes, while scope precision preserves unrelated work.

\section{RQ3: How readily do agents adopt the authority contract?}
\label{sec:adoption}
\begin{table}[t]
	\caption{All four independent-adoption cohorts and the pooled stage profile. (a) Each entry gives accepted/70 and total model-response cost in dollars. (b) Each entry has denominator 30 across the three benchmark-source tasks. Stages start independently.}
	\label{tab:adoption}
	\centering
	\begin{tabular}{@{}lrrrr@{}}
		\toprule
		\multicolumn{5}{l}{(a) Cohort outcome and fixed-workload cost}\\
		Source & WTB & Git/Auth & CONTINUITY & MasDrift \\
		\midrule
		\tableinput tables/adoption_cohorts_main.tex
		\bottomrule
	\end{tabular}
	\medskip
	
	\begin{tabular}{@{}lrrrr@{}}
		\toprule
		\multicolumn{5}{l}{(b) Three-source stage acceptance}\\
		Stage & WTB & Git/Auth & CONTINUITY & MasDrift \\
		\midrule
		\tableinput tables/adoption_stages.tex
		\bottomrule
	\end{tabular}
\end{table}

\paragraph{Independent stages isolate the integration burden.}
Source-Derived Adoption Cost assigns ten blocks, four methods, and seven stages to an earlier controlled fixture and to one fixed task from each of TeamBench, BulkPR-Bench, and Terminal-Bench. Each cohort contributes 280 slots, totaling 1,120. E0 covers evidence acquisition/concurrency; E1 publication; E2 recovery; E3 a second workflow; E4 contribution-chain handoff; E5 receipt-envelope change; E6 strict parsing. Each slot starts from its specified scaffold with the common within-study coding budget. Costs include failed and repaired model work.

\paragraph{The adoption advantage repeats across source tasks.}
WTB leads accepted-stage count and lowers fixed-workload cost in all four cohorts (Table~\ref{tab:adoption}a). Three-source transfer gives 176/210 WTB acceptances, versus 121/108/103; costs per acceptance are \$0.1313/\$0.2397/\$0.2954/\$0.3277. Its strongest stage advantages are E0 acquisition, E1 publication, E2 recovery, and E4 contribution handoff.

\paragraph{Publication tests the version-bound interface.}
The final E1 cohort accepts 19/30 WTB transfer slots, versus 7/30 Git/Auth, 0/30 \cont{}, and 3/30 \mas{}. WTB accepts 7/10, 7/10, and 5/10 on the three source tasks, plus 5/10 on the controlled source. Agents start from the supplied E0 predecessor and must bind approval to admitted material version, content, and current authority. Successful WTB integrations establish practical feasibility within budget; unsuccessful work remains counted. Git/Auth and \cont{} complete more E3 and E6 slots.

\paragraph{Bundles and cost answer narrower questions.}
WTB passes independent E0--E2 bundles in 13/30 blocks, versus 1/0/0, and full E0--E6 bundles in 9/30, versus zero for each comparator. E3--E6 bundles pass in 20/30 versus 10/10/7. WTB costs less in all seven joint E2 successes with Git/Auth, while Git/Auth costs less in most joint E3 successes. RQ5 tests a persistent candidate.

\section{RQ4: Does integration transfer to real repository code?}
\label{sec:swe}
\begin{table}[t]
	\caption{Prospective SWE integration: eight primary repositories under identical stage budgets. Counts give accepted/8; costs include previous paid work through the named stage. Later stages can repair earlier failures. Ten qualified repositories including two development repositories yield 120 terminal stages.}
	\label{tab:swe}
	\centering
	\begin{tabular}{@{}lrrrrr@{}}
		\toprule
		Method & Core C0 & Recovery E2 & Handoff E4 & Cost to E2 (\$) & Cost to E4 (\$) \\
		\midrule
		\tableinput tables/swe_main.tex
		\bottomrule
	\end{tabular}
\end{table}

\paragraph{Prospective qualification and cumulative acceptance.}
Twelve selected repositories yield ten qualified execution workloads: eight primary and two development repositories. Four arms begin with empty application integration code and carry paid code through core C0, recovery E2, and actual two-writer handoff E4. Acceptance combines native project tests with five, seven, and eleven cumulative governance scenarios. Earlier unsuccessful code remains available for repair, explaining why later acceptance can exceed core acceptance. The measured task uses project reference patches as material and requires a working governance integration.

\paragraph{Recovery and handoff transfer within budget.}
WTB accepts 5/8 core integrations, 8/8 recovery integrations, and 4/8 handoffs; each comparator accepts 1/8 recovery integrations and zero handoffs (Table~\ref{tab:swe}). Through recovery, WTB spends 16.4\%/33.7\%/7.7\% less than Git/Auth/\allowbreak\cont{}/\allowbreak\mas{}. WTB achieves 2/2, 1/2, and 1/2 on development repositories. Separate calibration yields 3/3 completed projects for WTB, 1/3 for Git/Auth, and zero for both research baselines. Appendix~\ref{app:swe} details repositories and charged continuation.

\paragraph{Cost tradeoffs and uncertainty.}
Joint-success costs limit the aggregate cost claim. Kafka recovery costs \$0.2533 with WTB versus \$0.4656 with \cont{}; Tornado recovery costs \$0.6305 versus \$0.2166 for Git/Auth and \$0.2188 for \mas{}. The first jointly completed calibration project likewise favors Git/Auth, \$0.2229 versus \$0.4725. Each primary recovery contrast has seven WTB-only and zero comparator-only successes; each handoff contrast has four and zero. Exact $p$ values are 0.015625 and 0.125, and six-comparison Holm values are 0.09375 and 0.375. The observed completion differences are large in this qualified sample; its small size leaves substantial uncertainty about broader repository populations.

\section{RQ5: Does the same candidate sustain successive changes?}
\label{sec:continuous}
\begin{table}[t]
	\caption{Continuous adaptation: each entry gives complete E3--E6 trajectories, followed by E4 accepted/entered in parentheses. Denominators retain observed cohorts and predecessor-dependent E4 entry.}
	\label{tab:continuousmain}
	\centering
	\begin{tabular}{@{}lrrrr@{}}
		\toprule
		Source & WTB & Git/Auth & CONTINUITY & MasDrift \\
		\midrule
		\tableinput tables/continuous_compact.tex
		\bottomrule
	\end{tabular}
\end{table}

\paragraph{Sequential adaptation changes the experimental unit.}
All twelve task--method E2 starts satisfy the shared initial contract. E3 adds a workflow, E4 a multi-contributor chain, E5 a receipt-envelope change, and E6 strict parsing. Each stage consumes the preceding accepted candidate; failure closes its trajectory. The complete-case analysis contains 117 observed trajectories, 468 stage outcomes, and 8,067 completed response records totaling \$27.6157. Costs and matched task--block comparisons condition on this cohort; Appendix~\ref{app:continuous} defines inclusion.

\begin{figure}[t]
	\centering\includegraphics[width=\linewidth]{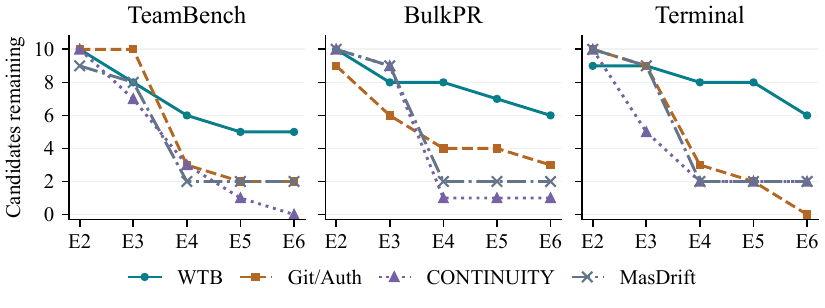}
	\caption{Observed candidate counts across E3--E6, starting from qualified E2. Counts follow each task--method cohort of nine or ten trajectories; Table~\ref{tab:continuousmain} gives exact denominators. WTB retains the most complete trajectories on each task, with substantial separation at E4 handoff.}
	\label{fig:outcomes}
\end{figure}

\paragraph{Completion advantage is sustained across all three tasks.}
WTB completes 5/10 TeamBench, 6/10 BulkPR, and 6/9 Terminal trajectories, exceeding each comparator on each task (Table~\ref{tab:continuousmain}). Figure~\ref{fig:outcomes} localizes the separation: Git/Auth passes E3 on all ten TeamBench trajectories but E4 on three; WTB passes 6/8 E4 entrants there, 8/8 in BulkPR, and 8/9 in Terminal. Maintaining multi-contributor authority is a distinct integration burden.

\paragraph{Handoff failures connect integration back to authority semantics.}
E4 requires earlier and later contributors' exact attempts, chain continuity, renewed authority at publication, and recovery of the same obligations. Recorded comparator failures include an earlier contributor's revoked/replaced execution and discontinuous material chains. WTB connects the DSH receipts and material references to its native publication path. Successful comparator trajectories satisfy the same checks and remain visible in the data. WTB also loses candidates at later changes: BulkPR declines from eight accepted E4 stages to six final successes, and Terminal from eight to six. The infrastructure supports reuse while application adaptation remains fallible.

\paragraph{Incremental cost and accepted code.}
WTB costs \$0.3944/\$0.3836/\$0.3544 per complete trajectory. All six joint-success pairs favor WTB by 1.16--54.32\%; Git/Auth spends less overall on BulkPR while completing fewer trajectories. Successful WTB integrations add mean application-code counts of 135/135/116 lines, versus 208--302 across successful comparator groups. These measurements quantify incremental work above supplied implementations and libraries.

\paragraph{Statistical and substrate boundaries.}
All nine continuous Holm-adjusted comparisons exceed 0.05, with minimum 0.28125. Shared-backend capability also ties: every method passes 6/6 checks per task; six closed-store same-backend moves recover twelve exact receipts without application changes. These JSONL/SQLite results concern DSH storage/restart behavior. Adapter-specific capabilities remain as defined in Appendix~\ref{app:adapters}.

\section{Related work}
\label{sec:related}
\paragraph{Agent authority and information flow.}
Complete mediation places authorization at access boundaries \citep{saltzer}. CaMeL separates control/data flow and checks capabilities; Fides enforces information-flow policies; PAuth binds tool calls to task-scoped intent using provenance envelopes \citep{camel,fides,pauth}. \cont{} composes effect-bound security context, while \mas{} studies preservation through delegation \citep{continuity,masdrift}. Our focus is authority tied to exact producing attempts, material versions, scoped generations, and durable recovery. The controlled intervention ladder separates information, representation, and enforcement; enhanced baseline parity locates the subsequent empirical question in integration and maintenance.

\paragraph{State and provenance infrastructure.}
MemGPT and AIOS organize context and execution services; LangGraph, Ray, and LogAct support persistent or distributed execution and recovery \citep{memgpt,aios,langgraph,ray,logact}. CoAgent coordinates concurrent mutation \citep{coagent}. in-toto binds authorized steps to materials/products, and PROV-O models provenance relations \citep{intoto,provo}. WTB joins historical execution evidence with mutable authority at an effect boundary. This extends the experimental focus from having and representing provenance to maintaining its authorization obligations during downstream use.

\section{Limitations and conclusion}
The reference semantics assumes trusted state and mediated effects, including receiver cooperation. Runtime coverage comprises named fault schedules and six reused captures; live coverage includes two runtimes. Coding results measure incremental integration from supplied libraries, on fixed tasks and recorded model routes. Human/library construction effort is unpriced. The E1 cohort establishes a version-bound publication advantage under this interface and budget; it does not isolate which API affordance accounts for each coding decision. Paired coding comparisons are small, and their adjusted significance tests remain inconclusive. Additional tasks, models, and effect APIs would test the reach of the measured advantages.

Authority is an execution-to-effect relation that must survive changes in time, contributors, and location. Our formalization identifies the information needed to decide it and the transitions needed to maintain it. The experiments show that facts resolve ambiguity, enforcement controls fixed intents, and complete responsibilities recover capability across architectures. WTB packages those responsibilities into infrastructure with higher observed completion in acquisition, publication, recovery, handoff, and continuous maintenance. The empirical claim is reusable authority infrastructure: strong observed integration alongside semantic parity of complete controls.

\label{main:end}

\section*{Reproducibility statement}
The companion experiment package contains compact observations, source provenance, protocols, reconstruction programs, and detailed reports for all six study families. Appendix~\ref{app:protocol} maps the manuscript to its numeric inputs. The companion table/figure builder loads saved records without model calls. All generated tables, vector figures, TeX sources, and bibliography needed to compile this manuscript accompany its source archive.

\section*{AI use statement}
LLMs served as experimental agents in the controlled planning and paid coding studies. AI tools assisted with finding related work, discussing research ideas, drafting proofs of the paper's propositions, and writing and editing the manuscript. The authors verified the reported results and take responsibility for the paper's claims.
	
\bibliographystyle{plainnat}
\bibliography{references}

\appendix
\section{Authority sufficiency and the observation boundary}
\label{app:aliasing}

\paragraph{State and observation.}
Let $\mathcal S$ be a domain of states $S=(W,E,A,L)$ and $\mathcal A$ a set of concrete effect requests. Material state $W$ includes admitted versions and contributor-chain relations. Evidence $E$ includes admitted execution identities and reconciled terminal receipts. Authority $A$ supplies scoped actor permissions and generations. Ledger $L$ records approval consumption, reservations, and effects. A request names its target and intended effect; admissibility is $Y:\mathcal S\times\mathcal A\to\{0,1\}$. The observation $O:\mathcal S\to\mathcal O$ may project or serialize these components.

\paragraph{Proof of Proposition~\ref{prop:sufficiency}.}
Suppose $d(O(S),a)=Y(S,a)$ exists. For any $S,S'$ with $O(S)=O(S')$, substitution gives $Y(S,a)=d(O(S),a)=d(O(S'),a)=Y(S',a)$. Conversely, suppose $Y(\cdot,a)$ is constant on every fiber $O^{-1}(o)$. Define $d(o,a)$ as that common value on each nonempty fiber; its value outside the observation image is arbitrary. Then $d(O(S),a)=Y(S,a)$ for all $S,a$. This establishes both directions.

The criterion depends on the state/request domain. An observation can suffice for one scope and fail for another; it can also suffice for the current state while omitting a later mutation. Our paired construction fixes material, ordinary Git-visible lineage, and visible dialogue, then changes the external actor--scope relation. It directly creates an observation fiber with two different required decisions.

\begin{proposition}[Authority aliasing]
For fixed request $a$, if $O(S_+)=O(S_-)$ and $Y(S_+,a)\ne Y(S_-,a)$, a decision rule using only $O$ assigns the same publication probability to both worlds. Under a uniform distribution over this pair, its expected decision accuracy is $1/2$.
\end{proposition}
\begin{proof}
Let $p$ be the shared publication probability. Correctness probabilities in the publish and hold worlds are $p$ and $1-p$; their mean is $1/2$.
\end{proof}

\paragraph{Information-preserving representations.}
Let $F(S)$ be the admitted authority facts and let raw and typed encoders $e_R,e_T$ have decoders satisfying $d_R(e_R(F))=F$ and $d_T(e_T(F))=F$ on the evaluated domain. Both encodings distinguish exactly the same values of $F$. Their information partitions therefore agree. An actual model $\pi$ can nevertheless satisfy $\pi(e_R(F))\ne\pi(e_T(F))$ because its computation depends on the sequence supplied. R0--W0 evaluates this behavioral difference under equal admitted facts. Sufficiency guarantees existence of a decision rule, while the experiment measures specific models.

\paragraph{Time and effect admission.}
An observation taken at $t_0$ can be fully sufficient for $Y(S_{t_0},a)$ while leaving $Y(S_{t_1},a)$ unresolved after a scope revocation at $t_1$. A retained generation distinguishes revoke-and-regrant from uninterrupted permission; a current-attempt pointer distinguishes a successful old execution from its replacement. Checking these facts at the effect boundary is the dynamic counterpart of observation sufficiency. The fixed-intent experiment holds $a$ constant and intervenes on admission; live interleavings mutate authority between execution admission and publication.

\section{Reference transition semantics and implementation responsibilities}
\label{app:protocolproof}

\subsection{Bindings, state, and assumptions}
For $b=(C,m,v,q)$, contributor $c_i\in C$ binds $(x_i,k_i,u_i,s_i,r_i,g_i)$. The predicate $\mathrm{Fresh}$ requires a current admitted attempt $k_i$ for $x_i$, a reconciled completed receipt with the matching actor/scope/input identity, and current authorization of $u_i$ in $s_i$ with unchanged generation $g_i$. Material-chain matching requires that the recorded sequence of contributions yields the exact admitted target/version. Approval matching binds that entire snapshot. Version is an admission identity: repeated identical bytes can occupy distinct versions.

The reference semantics assumes (i) faithful native capture and actor assignment; (ii) durable, trusted execution, material, and policy stores; (iii) complete mediation of the effects under discussion; (iv) serialization between relevant local state updates and check/consume/commit; and (v) a cooperating external receiver with authoritative lookup and atomic token validation plus a unique effect-key insertion. These assumptions describe the interfaces that an implementation must provide. Each study exercises a specified subset of the contract; the full predicate also appears in downstream application acceptance.

For each approval $q$, the ledger records consumption and at most one immutable reservation $(q,\text{sink},\text{payload},\text{token})$. The logical external states are ready, pending/unknown, applied, and rejected. ``Absent'' is a receiver observation used in recovery. It leaves the reservation intact. A stored receipt records an existing effect and does not grant permission for a new effect after a policy change.

\subsection{Permitted transitions}
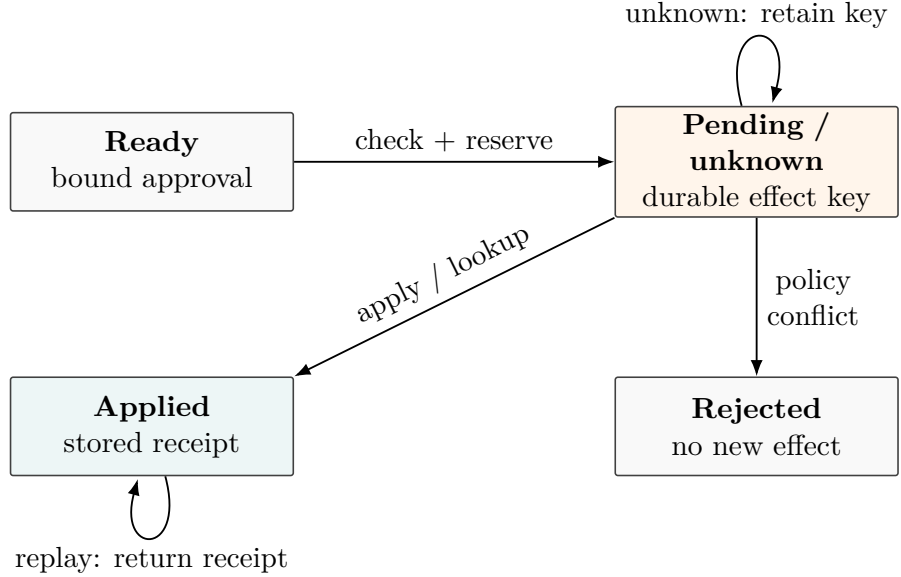
\begin{figure}[t]
\centering
\begin{tikzpicture}[x=1mm,y=1mm,font=\normalsize,>=Latex,
 box/.style={draw=black!75,line width=.6pt,rounded corners=1pt,align=center,text width=36mm,minimum height=13mm,inner sep=2pt},
 flow/.style={->,line width=.7pt}]
\node[box,fill=gray!5] (ready) at (21,0) {\textbf{Ready}\\bound approval};
\node[box,fill=orange!8] (pending) at (101,0) {\textbf{Pending / unknown}\\durable effect key};
\node[box,fill=teal!7] (applied) at (21,-35) {\textbf{Applied}\\stored receipt};
\node[box,fill=gray!5] (rejected) at (101,-35) {\textbf{Rejected}\\no new effect};
\draw[flow] (ready)--node[above,align=center]{check + reserve}(pending);
\draw[flow] (pending.south west)--node[above,sloped]{apply / lookup}(applied.north east);
\draw[flow] (pending)--node[right,align=center]{policy\\conflict}(rejected);
\draw[flow] (pending) edge[loop above] node[above]{unknown: retain key} (pending);
\draw[flow] (applied) edge[loop below] node[below]{replay: return receipt} (applied);
\end{tikzpicture}
\caption{External-effect lifecycle. Reservation checks authority and consumes approval atomically. A receiver validates its current authority token and deduplicates the key. Unknown outcomes retain the reservation; confirmed absence permits a checked retry. Local publication combines check, consumption, and effect in one transaction.}
\label{fig:contract}
\end{figure}

\begin{table}[htbp]
\centering
\caption{Reference operations and their atomic or durable responsibility. Each new effect is associated with one reserved identity.}
\begin{tabularx}{\linewidth}{@{}l X@{}}
\toprule
Operation & Required transition \\
\midrule
Capture & Bind immutable native evidence to the admitted execution/attempt; reconcile duplicates and terminal observations. \\
Approve & Read admissible contributor/material state and bind approval to its exact snapshot. \\
Local publish & In one transaction, check Equation~\ref{eq:contract}, consume approval, and insert the effect under its unique key. \\
Reserve external & In one transaction, check Equation~\ref{eq:contract}, consume approval, and persist the immutable reservation as unknown before sending. \\
Receiver apply & Atomically compare the supplied authority token with receiver authority and insert the key/payload; return a matching existing receipt on replay. \\
Lookup applied & Persist or return the receiver receipt; produce no additional write. \\
Lookup absent & Recheck contributor/material authority and approval snapshot against the reservation; retry the same key/payload/token when they still match. \\
Lookup unknown & Retain the reservation and defer further effect attempts until its outcome can be resolved. \\
Conflict & Reject a changed binding or conflicting receiver request; preserve any existing effect record. \\
\bottomrule
\end{tabularx}
\end{table}

After reservation, retry checks exclude $\mathrm{Unconsumed}$: the same approval has already been consumed by that reservation. Retry instead requires the stored reservation, an authoritative absence observation, and current contributor/material/approval matching. This preserves single consumption while allowing recovery. A delayed original sender can race with retry; receiver deduplication serializes their common key.

\subsection{Proof of Proposition~\ref{prop:safety}}
Initially the effect and reservation ledgers are empty. For a local effect, the only insertion transition checks Equation~\ref{eq:contract} in the same serialized transaction that consumes approval and inserts the effect. Relevant state writers share that serialization boundary. The predicate therefore holds in the atomic transition's pre-state; its post-state contains both the effect and the consumed approval. The unique effect key and consumption state prevent a second fresh insertion; an existing-result path returns its prior receipt.

For an external effect, the only initial send follows a committed, admissible reservation. All later sends reuse its immutable key, payload, and token. A retry additionally requires authoritative absence and a fresh local snapshot check. Independently of send timing, a receiver insertion atomically checks its current token. Thus every newly inserted external effect has a valid reservation and satisfies receiver authority at receiver commit. Its unique key prevents a second insertion, including under delayed senders or concurrent retry. Applied-result lookup preserves the count. Unknown-result retention and rejection insert no effects. Induction over these transitions proves the stated properties.

The two commit points have deliberately separate predicates. Receiver authority is the state checked by the receiver; a policy service must supply the appropriate token there. A global policy-update/remote-write ordering would additionally require coordination of those services. The reference proposition does not infer that ordering from a local pre-send check. Our external-effect tests instantiate the cooperating receiver explicitly.

Safety alone permits indefinite pending state when a receiver cannot answer. Recovery availability additionally needs eventual reliable lookup, stable admissible authority for retries, and a reachable sink. The lost-response scenarios test this conditional recovery: the effect already exists, so its receipt can be recovered without a second write.

\subsection{Implementation and experiment correspondence}
The WTB publication path supplies transactional publication and durable external delivery. Its application adapters connect native receipts, material observations, and version-sensitive approvals to the contract. The SQL guard implements the same tested decisions independently; the CONTINUITY-bound control places them in its context-bound path. All three share the receiver's token-check and idempotency services in runtime experiments. These shared services contribute to their common capability.

The five responsibility ablations map directly onto the state account. Exact-attempt removal weakens execution identity in $\mathrm{Fresh}$. Generation removal weakens temporal authority. Consumption removal weakens uniqueness of fresh publication. Monotone-reconciliation removal changes the meaning of terminal evidence under reordered/conflicting events. Unknown-recovery removal preserves conservative safety but loses the ability to recover an already committed remote effect. The observed signatures in Table~\ref{tab:ablations} test these individual obligations.

The final adoption E1 cohort tests the implementation boundary directly. A version-sensitive approval relation is required by the application contract; agents integrate it from a supplied evidence/acquisition predecessor. WTB accepts 24/40 publication slots across four sources, including 19/30 benchmark-source slots. The qualified E0--E2 starts for continuous adaptation establish a different, supplied starting condition. Together, these observations connect the reference semantics to both initial coding and subsequent maintenance without conflating their experimental units.

\section{Runtime adapters and architecture realization}
\label{app:adapters}

Figure~\ref{fig:architecture} presents the publication-relevant WTB architecture. Figure~\ref{fig:adapters} expands its runtime boundary. The runtime owns native execution and returns facts; WTB owns the admitted attempt and its use in an authority decision. The workspace owner supplies material observations, the authority owner supplies scoped policy, and the publication sink joins them transactionally. Application integrations connect the task's material/version relation to this core.

\begin{figure}[t]
\centering
\begin{tikzpicture}[x=1mm,y=1mm,font=\normalsize,>=Latex,
 box/.style={draw=black!75,line width=.6pt,rounded corners=1pt,align=center,inner sep=3pt,text width=38mm},
 flow/.style={->,line width=.7pt}]
\node[box,fill=teal!7,text width=130mm,minimum height=14mm] (head) at (69,0) {\textbf{WTB execution admission}\\persist exact session, provider identity, configuration, input, attempt\\then dispatch with the same binding envelope};
\node[box,fill=blue!4,minimum height=15mm] (dsh) at (21,-25) {\textbf{DSH adapter}\\prepare / invoke\\reconcile};
\node[box,fill=blue!4,minimum height=15mm] (lg) at (69,-25) {\textbf{LangGraph facade}\\prepare / invoke\\reconcile};
\node[box,fill=blue!4,minimum height=15mm] (ray) at (117,-25) {\textbf{Ray capture bridge}\\persist dispatch fence\\submit / collect};
\draw[flow] (head.south -| dsh.north)--(dsh.north);
\draw[flow] (head.south)--(lg.north);
\draw[flow] (head.south -| ray.north)--(ray.north);
\node[box,minimum height=24mm] (dsht) at (21,-53) {Python SDK\\JSON-RPC server\\Cordis composition\\agent/session pipeline};
\node[box,minimum height=24mm] (lgt) at (69,-53) {State adapter\\compiled StateGraph\\thread + checkpoint\\checkpoint persistence};
\node[box,minimum height=24mm] (rayt) at (117,-53) {Ray remote task\\job/task/worker IDs\\returned dispatch ID\\driver-owned fence};
\draw[flow] (dsh)--(dsht); \draw[flow] (lg)--(lgt); \draw[flow] (ray)--(rayt);
\node[box,fill=gray!5,minimum height=16mm] (dshe) at (21,-82) {Durable message\\terminal outcome\\exact session + attempt};
\node[box,fill=gray!5,minimum height=16mm] (lge) at (69,-82) {Checkpoint reference\\terminal outcome\\process-bound identity};
\node[box,fill=gray!5,minimum height=16mm] (raye) at (117,-82) {Task receipt\\terminal outcome\\experimental integration};
\draw[flow] (dsht)--(dshe); \draw[flow] (lgt)--(lge); \draw[flow] (rayt)--(raye);
\node[box,fill=teal!7,text width=130mm,minimum height=14mm] (collector) at (69,-108) {\textbf{Shared EvidenceCollector $\longrightarrow$ EvidenceStore}\\match receipt to admitted binding; preserve native reference and terminal state\\publication separately joins evidence with current authority};
\draw[flow] (dshe.south)--(dshe.south |- collector.north);
\draw[flow] (lge.south)--(collector.north);
\draw[flow] (raye.south)--(raye.south |- collector.north);
\end{tikzpicture}
\caption{Runtime-specific paths to shared evidence. DSH uses the public SDK/JSON-RPC seam; LangGraph uses its workflow facade and state adapter; Ray uses an experimental capture bridge. Native receipt references remain attached to the admitted binding. Appendix~\ref{app:adapters} specifies each path's recovery capabilities.}
\label{fig:adapters}
\end{figure}
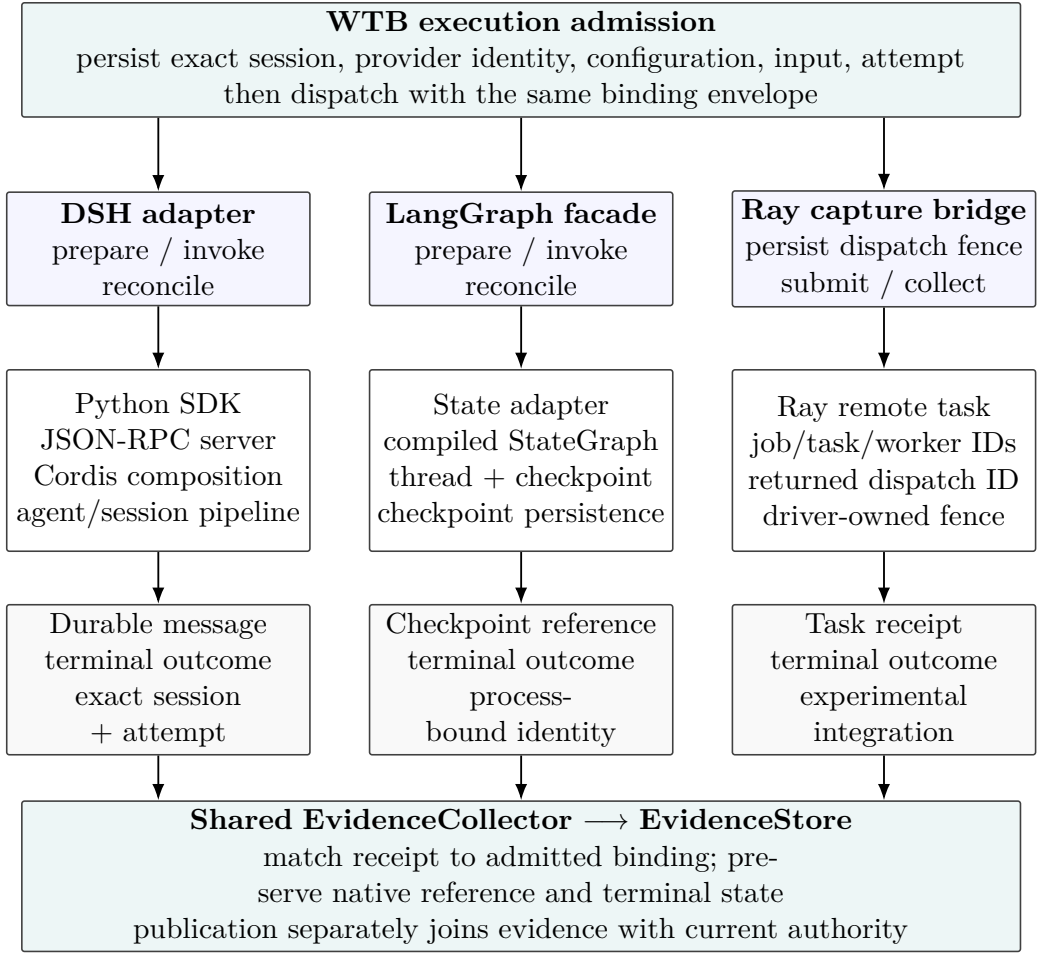

\paragraph{Common identity and receipt interface.}
The framework-neutral \code{WorkflowRuntime} port separates attempt preparation, invocation, and reconciliation. \code{RuntimeBindingEnvelope} contains exact session identity, opaque persistence-instance identity, runtime configuration identity, protocol/schema versions, input identity, and attempt ID. Preparation obtains these identities before work is dispatched. \code{RuntimeOutcomeReceipt} returns the matching envelope, terminal outcome, and native evidence. \code{EvidenceCollector} checks the execution, attempt, input, session, and complete binding before appending evidence to \code{EvidenceStore}. Native references remain available for interpretation: a DSH message, a LangGraph checkpoint, or a Ray task.

\paragraph{DSH/Cordis.}
DeepSeek Harness (DSH) uses Cordis for plugin composition \citep{dsh}. \code{DeepSeekHarnessRuntime} connects through its Python SDK to the JSON-RPC server. The runtime session and agent/tool pipeline remain DSH-owned. Preparation checks the SDK's provider identity and configuration before dispatch. Reconciliation uses the admitted session/attempt and durable native message evidence. The public SDK in the evaluated integration supports session receipts; its checkpoint-rewind capability is declared unavailable. This allows durable result lookup to have a precise meaning independent of checkpoint restoration.

\paragraph{LangGraph.}
\code{LangGraphWorkflowRuntime} wraps a state adapter that owns the compiled graph, sessions, and checkpoint persistence. The facade handles workflow binding, invocation, and bound outcomes. Checkpoints supply native references and support checkpoint-oriented operations. The evaluated compatibility facade uses a process-local opaque provider identity, so a newly created facade cannot establish identity continuity merely from the same checkpoint path. An old binding is rejected when that identity cannot be established. The experiments test the corresponding capture and publication responsibilities; they do not imply uniform process-restart recovery across adapters.

\paragraph{Ray prototype.}
The capture bridge persists a driver-owned dispatch fence, invokes a real remote task with native retries disabled, and checks returned job/task identity and the echoed dispatch binding before constructing its outcome receipt. It also records worker identity and the returned object reference. This third-runtime integration exercises native capture and the common authority core. Driver-restart reconciliation is outside the prototype's tested capabilities. Its role in the 960-case matrix remains separate from the two-runtime live interleaving study.

\paragraph{Engineering separation and evidence.}
WTB supplies a shared identity/evidence path and reusable publication/recovery services while preserving each runtime's invocation and persistence semantics. Exact-attempt, generation, consumption, and reconciliation ablations target that shared boundary. The coding studies measure application work above the supplied libraries, including material/version and multi-contributor adapters. Complete SQL and CONTINUITY-bound controls demonstrate the same tested contract at different implementation placements. WTB's observed advantage is the reuse of its packaged responsibilities during integration and maintenance.

\clearpage
\section{Experimental inventory and interpretation}
\label{app:inventory}

The experimental record spans six families. We count a substudy as a separately specified protocol or source cohort: three original-mechanism components, one deterministic runtime matrix, two live runtime cohorts, four independent-adoption cohorts, two SWE integration studies, and three continuous-adaptation cohorts. The equal-contract comparison supplements the original mechanism family. Source functional checks and W1 task replays support workload provenance. They remain separate from the formal study totals.

Four benchmark families contribute workload material. The executed workload count comprises ten qualified prospective SWE repositories, three calibration projects, and three fixed source tasks. The earlier controlled adoption cohort contributes a governance fixture. The 12 selected prospective SWE repositories yield ten qualified execution workloads; the two unqualified selections contribute no measured method outcomes.

\begin{table}[htbp]
\centering
\caption{Evidence modes and counting units. A coding stage can contain many model responses and acceptance checks. Counts across these rows must retain their unit labels.}
\begin{tabularx}{\linewidth}{@{}l X r@{}}
\toprule
Evidence mode & Unit and coverage & Count \\
\midrule
LLM mechanism & Authority cells across four arms & 128 \\
LLM planning & Saved proposals across three arms & 96 \\
Deterministic enforcement & Matched saved W0 intents & 32 \\
Contract coverage & Seven implementations, 56 checks each & 392 \\
Deterministic runtime & Six native captures forked across worlds/arms & 960 \\
Live runtime & In-flight authority change executions & 24 \\
LLM independent adoption & Four source cohorts, seven stages & 1,120 \\
LLM prospective SWE & Ten repositories, four arms, three stages & 120 \\
LLM continuous adaptation & Complete observed trajectories & 117 \\
Continuous stage outcomes & Four stages per observed trajectory & 468 \\
\bottomrule
\end{tabularx}
\end{table}

The two controlled matrices contain $128+960=1{,}088$ cases at their respective cell definitions. Independent adoption, prospective SWE, and continuous adaptation contribute $1{,}120+120+468=1{,}708$ coding-stage outcomes. These arithmetic totals provide an inventory; they do not define an inferential sample. In particular, the continuous study has three source tasks, 117 observed trajectories, and 8,067 completed response records. Task replication, trajectory replication, and within-trajectory work answer different questions.

\paragraph{Study names.}
Runtime Responsibility Ablation names the deterministic cross-runtime responsibility matrix. Live Authority Interleaving Validation names the in-flight authority-change study. Source-Derived Adoption Cost names independent integration stages. Continuous Adaptation Trajectories names sequential E3--E6 maintenance. Historical identifiers survive only in the companion data's path mapping so existing files remain discoverable.

\paragraph{Three comparison levels.}
At the observation level, G0/GM0, C0, R0, and W0 vary evidence exposed to the model. At the implementation level, Git/Auth, \cont{}, \mas{}, and WTB differ in their supplied governance mechanisms. Original and enhanced variants are evaluated separately for contract completeness. At the execution level, DSH/Cordis, LangGraph, and Ray provide runtime evidence. SQL guard and \cont{}-bound adapter are complete effect-control implementations. Keeping these levels explicit prevents a runtime's persistence feature from being treated as an authorization algorithm.

\subsection{Supplied capabilities and application starting points}
\label{app:starts}

The paid arms identify a native package together with a study-specific application start. The original/enhanced labels in Table~\ref{tab:parity} instead identify configurations in the separate 14-scenario contract comparison. That table establishes achievable coverage. The paid experiments measure incremental work above the packages and starts specified here. Responsibility coverage, API convenience, and application wiring jointly contribute to their observed differences.

\paragraph{Native capability inventory.}
All arms can inspect their assigned read-only native source and public API/contract documentation, write application modules and self-tests, and invoke the assigned native mechanisms. The shared DSH runtime supplies execution records. WTB exposes \code{EvidenceCollector}, attempt identity and persistent evidence stores, plus publication, external-delivery, and authority operators. Git/Auth exposes real Git object storage and compare-and-swap/transactional reference updates; application code supplies authorization, receipt mapping, and effect/recovery integration. \cont{} exposes grants, security contexts, envelopes, proof verification, policy epochs, execution permits, nonce ledgers, finality sinks, and effect brokers. \mas{} exposes policy compilation, scoped stage state, and native authorization monitoring. The latter two retain their native denial checks; DSH evidence persistence and application mapping remain integration work. This inventory describes supplied functionality, without assigning equal implementation effort to the four packages.

\paragraph{Independent-adoption starts.}
Each arm's E0 starts with empty application code. E1 starts from a frozen evidence/acquisition implementation; the agent adds publication. E2 starts from a publication predecessor; the agent adds external recovery. E3 and E4 start from recovery-capable references and add a workflow and contribution-chain handoff, respectively. E5 also receives chain support and adds the transport envelope; E6 receives that envelope and adds strict parsing. The predecessor builder copies arm-specific reference modules and exposes the common application interface. Predecessors satisfy their supplied obligations while leaving the target change to the agent. Their historical construction is outside the measured work. Each slot is initialized independently, so an E2 success can coexist with an E1 failure in the same block.

\paragraph{Prospective SWE starts.}
Each arm/repository begins with empty candidate and project-governance source. The supplied project reference patch supplies task material; the agent receives no reference governance implementation or candidate repair. Read-only native packages remain available. The agent implements the public application interface in the project's governance module, connects actual receipts and material, and retains its code through C0, E2, and E4, including unsuccessful work. Acceptance combines project tests, cumulative governance behavior, and native-mechanism invocation. The separate calibration study permits reusable charged application code across its three projects and retains its own continuation protocol.

\paragraph{Continuous-adaptation starts.}
All four arms receive enhanced, arm-specific reference application implementations from the same source-task setup. Each of the twelve task--arm starts passes the common E0--E2 qualification before paid maintenance begins. This includes version-sensitive publication and recovery. Starting capability is therefore qualified at a common boundary, while code and APIs retain their architectural differences. The agent adds E3--E6 behavior and carries only its preceding accepted candidate forward. These supplied starts establish the initial condition for incremental maintenance; the independent paid E1 cohort measures coding publication from an earlier predecessor.

\paragraph{Budgets and reproducible inputs.}
Independent adoption, prospective SWE, and continuous adaptation use at most 64 valid responses, 64 requested actions, and three acceptance submissions per stage, with DeepSeek V4 Flash at high reasoning and a 16,384-token response cap. Native/public sources remain read-only; generated application code and tests are editable. The compact companion package documents independent-adoption starts and stage contracts in \path{WTB/source-derived-adoption-cost/REPORT.md} and each source's \path{DETAILED_REPORT.md}. The final publication records and reconstruction entry point are in \path{submission/e1-final/}. The prospective study's \path{WTB/benchmark-experiment/REPORT.md} describes empty starts, cumulative C0--E2--E4 carryover, and primary/development cohorts. The continuous study's \path{submission/v23-continuous/PROTOCOL.md} specifies qualification and sequential carryover. These records distinguish supplied mechanisms from measured application work.

\section{Original Mechanism: evidence, intent, and effect}
\label{app:original}

\subsection{Question and controlled design}
The first question concerns observability: can authority evidence resolve cases in which ordinary artifact information leaves the publication obligation ambiguous? Four semantic templates, two model routes, two lineage directions, and two authority worlds form 32 matched strata under each of four arms. Models are DeepSeek V4 Flash and GPT-5.4 Mini. G0 supplies Git evidence with candidate material unavailable; C0 adds both candidate generations and neutral context; R0 adds raw authority receipts; W0 encodes the same admitted primitives in typed form. C0--R0 isolates information and R0--W0 isolates representation. Final semantic scoring examines the completed artifact/decision. First-action correctness uses the prescribed policy for each arm and measures within-arm protocol adherence. ITT additionally requires adherence across the execution.

\begin{table}[htbp]
\centering
\caption{Authority-control outcomes. Every arm has 32 terminal cells. First-action correctness is arm-policy-specific; final semantic success is the common outcome for the main information comparison.}
\begin{tabular}{@{}lrrrr@{}}
\toprule
Arm & First action & Semantic success & ITT composite & Unsafe adoption \\
\midrule
G0 & 32/32 & 0/32 & 0/32 & 0/32 \\
C0 & 0/32 & 0/32 & 0/32 & 0/32 \\
R0 & 32/32 & 32/32 & 31/32 & 0/32 \\
W0 & 32/32 & 32/32 & 32/32 & 0/32 \\
\bottomrule
\end{tabular}
\end{table}

The R0--C0 semantic contrast has 32 favorable and zero unfavorable discordances. The exact two-sided probability is $2/2^{32}=4.6566\times10^{-10}$ for this finite paired matrix. R0 and W0 tie on semantic success. Thus, access to authority facts explains this contrast. G0 follows its prescribed first action in all 32 cells while its unavailable candidate prevents final semantic success. Its first-action score answers a different question from the common semantic outcome. R0's single ITT miss is a permitted-file reread that preserved final semantics while violating the strict execution protocol.

\paragraph{Worked pair: log redaction.}
The saved \code{log-redaction-v362-01}, ADVANCE template starts with authorization-header redaction enabled and cookie redaction disabled. Candidate $g_1$ preserves that policy; candidate $g_2$ enables both. Git records $g_1$, while a later native attempt produces $g_2$. In the completed-publication world, its accepted external publication receipt makes $g_2$ authoritative: the correct resolving output is \code{ADOPT\_SOURCE} naming $g_2$. In the rejected-publication world, $g_1$ remains authoritative and already matches the target: the correct resolving output is \code{NOOP}. The observable target bytes and candidate meanings alone leave the external publication relation unresolved.

\paragraph{Raw and typed observations for this pair.}
C0 exposes both candidate implementations and their contract materials, together with Git evidence and the accepted-publication-channel policy, but omits the external relation. R0 adds \code{rawExternalReceipts}: the exact runtime attempt and input, produced artifact, and publication request/response with its receipt or rejection. W0 organizes those same admitted primitives under \code{wtbRelation}, including runtime-attempt, generation-input, artifact, and publication-exchange fields. Neither condition needs the model to infer publication from the candidate's desirable content. In the saved ADVANCE rows, both tested models choose \code{BLOCK} under C0 in both worlds. Under R0 and W0 they adopt $g_2$ in the completed world and choose \code{NOOP} in the rejected world.

\paragraph{Common semantic success and safe abstention.}
The recorded task-level predicate requires permitted execution, a resolving action other than \code{BLOCK}, and a final material value matching the authoritative record. \code{NOOP} resolves a verified already-correct target; \code{BLOCK} leaves the authority question open. Thus the C0 rejected-world row retains the correct bytes but receives zero semantic-success credit because it abstains from resolution. Its unsafe-adoption outcome also remains zero. This explains the 0/32 C0 completion count: safe abstention and completed resolution are separate outcomes. Proposition~\ref{prop:sufficiency}'s balanced-pair corollary concerns a forced binary admissibility decision from identical observations. The experiment additionally permits abstention and scores a resolved artifact, so that binary accuracy statement is not a prediction of its completion rate.

\paragraph{Arm-specific first-action policy.}
The protocol prescribes \code{BLOCK} for G0's unavailable-candidate condition. C0, R0, and W0 retain the world-specific resolving target: adopt the authoritative candidate when it differs from the target, otherwise \code{NOOP}. C0 supplies candidate bytes but omits the external authority facts needed to resolve the paired worlds. These scores measure adherence to each assigned target; the common semantic predicate and separate unsafe-adoption metric carry the cross-arm outcome interpretation. The four templates cover request tracing, write idempotency, log redaction, and webhook-signature policy; the same scoring rule applies to each.

\subsection{Planning and matched-intent enforcement}
The shared-workspace planning study uses DeepSeek V4 Flash and GPT-5.4 Nano. Each observation arm receives 32 calls, with 16 unauthorized and 16 authorized situations. Pooled unsafe proposals are 12/16 under GM0, 9/16 under raw receipts, and 6/16 under W0. Unknown or invalid decisions are 2/32, 2/32, and 11/32, respectively. DeepSeek contributes unsafe counts of 5/8, 8/8, and 0/8 across these arms; GPT-5.4 Nano contributes 7/8, 1/8, and 6/8. All eleven W0 unknowns come from DeepSeek. The aggregate direction therefore combines a model-specific behavioral reversal with limited planner availability. This analysis bounds any general claim about typed representation improving model decisions.

The enforcement intervention takes the same 32 saved W0 planning outputs and applies the native publication gate. All 16 unauthorized situations produce zero unauthorized effects. In particular, the six unsafe publish intents remain unchanged inputs and are rejected at the effect boundary. All 12 eligible authorized intents execute. Four authorized-world rows contain no valid publish intent. The gate preserves authorized execution when such an intent exists; it cannot supply a missing planner action. Among the six discordant unsafe intents, the exact paired two-sided value is $0.03125$. This is a matched replay result over saved intents.

\subsection{Original and enhanced comparison mechanisms}
Four tasks in each of 14 scenarios yield 56 checks per implementation. Original Git/Auth passes 36, original \cont{} 44, and original \mas{} 28. Their enhanced implementations and WTB each pass all 56, yielding 392 check outcomes in total. Enhancement supplies missing execution, material, evidence, generation, and publication obligations as appropriate to each architecture.

The enhanced controls establish contract implementability across architectures. Reusable packaging and measured coding adoption are the subsequent empirical questions once complete contract coverage is supplied. Engineering hours required to create the enhancements were not measured. This family reports controlled information and effect behavior; coding cost is evaluated in the adoption and maintenance studies.

\section{Runtime Responsibility Ablation: fault-specific behavior}
\label{app:runtime}

\subsection{Design and responsibilities}
Native captures from DSH/Cordis, LangGraph, and Ray are evaluated on SQLite and PostgreSQL. Two captured attempts per runtime supply six native executions. Each runtime/backend pair evaluates 20 prescribed worlds under eight implementations, producing 960 deterministic cases. Worlds draw on the captured attempts and change controlled event, material, or authority inputs. This construction isolates contract behavior from repeated model or task execution. The recorded environment includes LangGraph 1.0.6 and Ray 2.49.2.

The three complete controls are WTB, an independently implemented SQL guard, and a \cont{}-bound adapter. Each passes 120/120 cases, executes all 60 expected legal effects, and produces zero unsafe effects. Five WTB ablations remove exact-attempt binding, monotone event reconciliation, scoped generation binding, single consumption, or unknown-effect reconciliation. Table~\ref{tab:ablations} reports the resulting outcomes.

\begin{table}[htbp]
\centering
\caption{Twenty controlled worlds, grouped by the publication obligation. ``One'' denotes one valid durable effect; ``zero'' denotes safe rejection.}
\begin{tabularx}{\linewidth}{@{}X l@{}}
\toprule
World or event sequence & Expected effects \\
\midrule
Clean execution; unrelated policy change & One each \\
Relevant policy change; revoke then regrant & Zero each \\
New admitted attempt; late old attempt & Zero each \\
Duplicate terminal; late running; new delivery identifier & One each \\
Running delivered before terminal & One \\
Conflicting completion delivered last; conflicting terminal states & Zero each \\
Failed terminal; missing terminal; material substitution & Zero each \\
Replay; concurrent replay & One total each \\
External effect with lost response; crash before external send & One total each \\
External retry with stale authority & Zero \\
\bottomrule
\end{tabularx}
\end{table}

\subsection{Failure interpretation}
Removing exact-attempt binding yields 108/120 correct outcomes and 12 unsafe effects. A valid receipt from a different or superseded execution cannot authorize the current publication. Removing scoped generation binding yields 114/120 and six unsafe effects: a previously valid actor/scope relation survives in the application after its governing authority changes. Removing single consumption yields 108/120 and 12 unsafe effects, locating replay safety in durable effect/approval state. In that ablation, all 60 legal-world rows set their completion flag, while only 48 satisfy the expected effect count and full contract. The main table reports the latter legal-correct metric explicitly.

Removing monotone reconciliation yields 102/120 correct outcomes, six unsafe effects, and 48/60 legal effects. This responsibility affects both safety and availability: later delivery order must preserve established lifecycle meaning while contradictory terminal evidence triggers reconciliation. Removing unknown-effect recovery yields 114/120 correct outcomes, zero unsafe effects, and 54/60 legal effects. Conservative rejection protects safety in these cases but loses valid recoverable work.

The matrix consequently separates responsibility-specific failure signatures. A single aggregate success number would conceal whether a removal admits an unauthorized write, duplicates an effect, or suppresses legal progress. Equal complete-control results establish the tested contract's portability across adapter and storage combinations. They do not measure the prevalence of each fault in production, independent replication over 960 fresh executions, or the comparative cost of maintaining the adapters.

\section{Live Authority Interleaving Validation}
\label{app:live}

\subsection{Temporal intervention}
This study asks whether a permission change committed during execution is respected at publication. The observed order is: admit the attempt, commit initial authority, enter native execution, commit the intervening authority change, receive the result, capture its native receipt, and decide publication. The two interventions affect either the active publication scope or an unrelated scope.

Two runtime cohorts, DSH/Cordis and LangGraph, each cross two backends, three complete controls, and two authority interventions. The resulting 24 rows are live native executions with deterministic payloads and no paid model calls. Each of the 12 relevant-scope changes yields safe rejection and zero effects. Each of the 12 unrelated changes preserves authority and yields one effect.

\begin{table}[htbp]
\centering
\caption{Live cohort results. Within each runtime, each implementation passes 4/4 and each backend passes 6/6.}
\begin{tabular}{@{}lrrrr@{}}
\toprule
Runtime & Total correct & Relevant correct & Unrelated correct & Total effects \\
\midrule
DSH/Cordis & 12/12 & 6/6 & 6/6 & 6 \\
LangGraph & 12/12 & 6/6 & 6/6 & 6 \\
\bottomrule
\end{tabular}
\end{table}

\subsection{Mechanism and boundary}
The result connects authority freshness to temporal order. A pre-execution check alone could preserve an obsolete permission through a long execution. All three complete implementations re-evaluate the binding after the intervening change. The unrelated intervention measures scope precision: a global change elsewhere should leave the current publication valid.

WTB, SQL guard, and the \cont{} adapter each pass 8/8. Thus, fresh authority evaluation succeeds under all tested complete implementations. Ray appears in the deterministic responsibility study; this live design covers DSH/Cordis and LangGraph. Longer in-flight intervals, multiple concurrent authority writers, crashes during the intervention, and further adapters require additional experiments. With all rows specified by design, counts describe the tested schedules directly.

\section{Source-Derived Adoption Cost: stage and source analysis}
\label{app:adoption}

\subsection{Question, intervention, and stage meanings}
The independent-adoption study measures coding work required to realize a common acceptance contract through four supplied mechanisms. All four arms run on DSH. Each task--block--arm--stage slot starts from its specified scaffold and receives the study's common coding budget. A later stage's start is independent of the preceding slot's accepted code. The design isolates stage-specific integration difficulty.

The seven stages are E0 evidence acquisition and concurrency, E1 publication, E2 recovery, E3 introduction of a second workflow, E4 contribution-chain handoff, E5 receipt-envelope change, and E6 strict parsing. Ten blocks per source, four arms, and seven stages give 280 slots per source cohort. The earlier controlled source cohort and three benchmark-source cohorts together yield 1,120 slots.

\subsection{Earlier controlled cohort}
WTB accepts 51/70 stages at \$7.4660, compared with 39/70 at \$10.2273 for Git/Auth, 37/70 at \$11.5691 for \cont{}, and 39/70 at \$11.7276 for \mas{}. Costs per accepted stage are \$0.1464, \$0.2622, \$0.3127, and \$0.3007. Acceptance repairs number 37, 69, 72, and 70. The outcome combines more accepted stages and less model spending for the same assigned workload.

Publication acceptance is 5/10 for WTB and 1/10 for each comparator. Recovery acceptance is 7/10 for WTB, 3/10 for Git/Auth, and 1/10 for each research baseline. Handoff acceptance is 8/10, 6/10, 2/10, and 3/10. These stage contrasts motivate transfer to fixed tasks from three benchmark sources.

\subsection{Three-source transfer and heterogeneous stages}
\begin{table}[htbp]
\centering
\caption{Independent-adoption results by source. Costs include accepted and method-failed work.}
\begin{tabular}{@{}llrrr@{}}
\toprule
Source & Method & Accepted & Cost (\$) & Cost/accepted (\$) \\
\midrule
\tableinput tables/adoption_by_source.tex
\bottomrule
\end{tabular}
\end{table}

WTB leads accepted-stage count and cost per accepted stage on all three tasks. TeamBench contributes 60 WTB successes, BulkPR-Bench 58, and Terminal-Bench 58, each out of 70. The pooled result is 176/210, compared with 121, 108, and 103 for the three comparators. Stage-level inspection localizes the advantage to acquisition E0, publication E1, recovery E2, and contribution handoff E4. Git/Auth and \cont{} complete more E3 and E6 slots, preserving a heterogeneous stage profile.

\begin{table}[htbp]
\centering
\caption{Per-source stage acceptance; each entry has denominator ten.}
\begin{tabular}{@{}lrrrrrrr@{}}
\toprule
Method & E0 & E1 & E2 & E3 & E4 & E5 & E6 \\
\midrule
\multicolumn{8}{l}{\emph{TeamBench}}\\
\tableinput tables/adoption_task1.tex
\midrule
\multicolumn{8}{l}{\emph{BulkPR-Bench}}\\
\tableinput tables/adoption_task2.tex
\midrule
\multicolumn{8}{l}{\emph{Terminal-Bench}}\\
\tableinput tables/adoption_task3.tex
\bottomrule
\end{tabular}
\end{table}

\paragraph{Publication version binding.}
The final E1 cohort contains 160 terminal slots, 40 per source and 40 per method. Across the three benchmark-source tasks, WTB accepts 7/10 TeamBench, 7/10 BulkPR-Bench, and 5/10 Terminal-Bench publication slots. Git/Auth accepts 2/10, 2/10, and 3/10; \cont{} accepts 0/10 on each; \mas{} accepts 0/10, 2/10, and 1/10. The earlier controlled source adds five WTB acceptances and one for each comparator. Thus the observed coding result supports a practical advantage for WTB's packaged version-bound publication interface under this fixed budget. The unsuccessful method outcomes remain in the denominator and cost accounting. Successful integrations exercise the contract's material-version relation, which distinguishes byte-identical material re-admitted under a new version.

\paragraph{Maintenance bundles and cost.}
The count of blocks whose independent E0--E2 slots all pass is 13/30 for WTB, 1/30 for Git/Auth, and zero for \cont{} and \mas{}. Full E0--E6 bundle counts are 9/30 for WTB and zero for each comparator. E3--E6 maintenance bundles pass in 20/30, 10/30, 10/30, and 7/30 blocks. These bundles aggregate separately initialized slots. A persistent candidate's survival through four changes is measured in Appendix~\ref{app:continuous}. Jointly accepted slots also yield heterogeneous cost: WTB uses less model-response cost in all seven jointly accepted E2 slots with Git/Auth, while Git/Auth uses less in most jointly accepted E3 slots.

The fixed-workload spending advantage describes these particular tasks, libraries, scaffolds, and budgets. It excludes construction of the starting libraries and human engineering effort. A low cost per accepted stage can reflect both cheaper attempts and more accepted attempts. Matched joint-success costs address only the overlapping successful subset, whose composition can itself depend on the method.

\section{SWE Repository Integration: transfer and recovery}
\label{app:swe}

\subsection{Selection and cumulative protocol}
The prospective selection contains 12 repositories. Ten satisfy input and native-test qualification before method outcomes are measured. Eight form the primary cohort: xarray-dataclasses, fsspec, crossplane, kafka-python, dlx, azure-cli, numba, and tornado. Plopp and Modin form the development cohort. Each qualified repository receives four method integrations and three stages, yielding 120 scheduled stage slots.

The coding agent begins with empty application integration code and supplied native libraries and project reference material. Core C0 combines native project tests with five governance scenarios. Recovery E2 extends the cumulative contract to seven scenarios. Two-writer handoff E4 extends it to eleven. Each stage permits at most 64 valid model responses, 64 actions, and three acceptance submissions. DeepSeek V4 Flash runs with high reasoning. Paid code carries forward, including unsuccessful earlier work, so a later stage can repair earlier obligations. Consequently, recovery counts may exceed core counts.

\subsection{Primary outcomes and paired uncertainty}
WTB accepts 5/8 C0 integrations, 8/8 recovery integrations, and 4/8 handoffs. Each comparator accepts 1/8 recovery integrations and zero handoffs. The six planned WTB contrasts comprise recovery and handoff versus each comparator. At recovery, each contrast has seven WTB-only and zero comparator-only successes, giving exact $p=0.015625$. At handoff, each has four and zero, giving $p=0.125$. Holm adjustment across the six yields $0.09375$ for each recovery contrast and $0.375$ for each handoff contrast. No contrast passes a familywise 0.05 threshold.

The effect sizes are useful descriptive evidence despite limited inferential power: WTB completes recovery on all eight qualified primary repositories, and handoff on half. A larger prospective repository sample would narrow the uncertainty and test sensitivity to selection and model configuration.

\subsection{Cost and counterexamples}
Through recovery, the eight-repository totals are \$3.1825 for WTB, \$3.8082 for Git/Auth, \$4.8002 for \cont{}, and \$3.4498 for \mas{}. WTB's fixed-workload spending is 16.4\%, 33.7\%, and 7.7\% lower, respectively. Through handoff, WTB costs \$6.0202; Git/Auth costs less at \$4.9509 while accepting no handoff. This illustrates how cost totals depend on attained outcomes and continuation behavior.

\begin{table}[htbp]
\centering
\caption{Jointly accepted primary project/stage cost pairs. Costs include earlier paid work through the indicated stage.}
\begin{tabular}{@{}llrr@{}}
\toprule
Repository/stage & Comparator & WTB (\$) & Comparator (\$) \\
\midrule
kafka-python C0 & CONTINUITY & 0.1803 & 0.1141 \\
dlx C0 & MasDrift & 0.1631 & 0.1519 \\
tornado E2 & Git/Auth & 0.6305 & 0.2166 \\
kafka-python E2 & CONTINUITY & 0.2533 & 0.4656 \\
tornado E2 & MasDrift & 0.6305 & 0.2188 \\
\bottomrule
\end{tabular}
\end{table}

Joint-success costs favor WTB in the Kafka recovery comparison and favor the comparator in the other four listed pairs. The conclusion is therefore higher fixed-budget completion with task-dependent cost tradeoffs. Native tests and governance checks jointly define acceptance. Since project artifacts originate from reference patches, autonomous issue-solving performance is outside this integration estimand.

\subsection{Calibration and development}
The earlier three-project calibration supplies one equally budgeted, charged continuation after all four arms fail the first repository under the initial budget. WTB eventually completes 3/3 at \$0.6641; Git/Auth completes 1/3 at \$0.3918 across attempted work; \cont{} and \mas{} complete 0/3 at \$0.3957 and \$0.2753. On the first jointly completed project, WTB spends \$0.4725 and Git/Auth \$0.2229. This directly bounds the cost claim.

On the two prospective development repositories, WTB accepts 2/2 core, 1/2 recovery, and 1/2 handoff stages at cumulative \$2.3250. Git/Auth and \cont{} accept none at \$1.7950 and \$1.7909. \mas{} accepts one core stage and no later stage at \$1.4564. Development outcomes document protocol behavior separately from the eight-repository primary analysis.

\section{Continuous Adaptation Trajectories: detailed source analysis}
\label{app:continuous}

\subsection{Sequential design and outcome definition}
Each of the twelve task--method starting implementations satisfies the common E0--E2 contract before paid E3--E6 adaptation. All arms share workload, model configuration, public contract, action budget, and finish budget. Native libraries are read-only. Application integrations may use their documented interfaces. Each stage allows up to 64 responses, 64 actions, and three finish attempts.

E3 adds a second workflow. E4 adds exact-attempt multi-contributor material-chain governance. E5 changes the receipt envelope. E6 adds strict parsing and maintenance requirements. A later stage consumes the preceding accepted candidate. Failure closes the remaining stages for that trajectory. The principal outcome is complete E3--E6 success; stage acceptance, code growth, and incremental model-response cost describe how that outcome arises.

The analyzed cohort contains 117 fully observed trajectories and 468 stage outcomes. Its 8,067 completed model-response records total \$27.615736318. Task--method denominators are nine or ten, as specified in Table~\ref{tab:continuousfull}. Inclusion requires an observable method outcome across the trajectory, including explicit predecessor-blocked stages after a method failure. Operationally interrupted trajectories with unresolved outcomes are outside this complete-case analysis; the original run protocol prohibits selective replacement. Filtering retains all observed method failures and their recorded response costs. Reported spending is conditional on the analyzed cohort and completed response records; total billing for every scheduled call is outside this estimand. Paired analyses use the corresponding observed task--block intersection. The operational archive is separate from the publication dataset. The complete-case comparison does not establish outcomes for unobserved executions or eliminate possible selection effects.

\begin{table}[htbp]
\centering
\caption{Complete trajectories and model-response cost. Undefined cost per success is marked with a dash.}
\label{tab:continuousfull}
\begin{tabular}{@{}llrrr@{}}
\toprule
Source & Method & Success & Cost (\$) & Cost/success (\$) \\
\midrule
\tableinput tables/continuous.tex
\bottomrule
\end{tabular}
\end{table}

\subsection{Stage survival and source-specific findings}
Figure~\ref{fig:outcomes} in the main text displays stage survival. The following table preserves each entering denominator, which changes after a predecessor failure.

\begin{table}[htbp]
\centering
\caption{Accepted/entered stage counts. Predecessor-blocked trajectories leave the denominator of subsequent entered-stage rates.}
\begin{tabular}{@{}llrrrr@{}}
\toprule
Source & Method & E3 & E4 & E5 & E6 \\
\midrule
\tableinput tables/continuous_stages.tex
\bottomrule
\end{tabular}
\end{table}

\paragraph{TeamBench.}
WTB completes 5/10 trajectories, Git/Auth 2/10, \cont{} 0/10, and \mas{} 2/9. Git/Auth has strong E3 entry performance, but only three of ten entering E4 candidates satisfy contribution-chain handoff. WTB accepts six of eight entrants at E4 and carries five through E5 and E6. The data locate a substantial separation at the handoff change. WTB's total cost is \$1.9719 and cost per completed trajectory \$0.3944, compared with \$1.2876 for Git/Auth and \$1.3677 for \mas{}. The single jointly completed Git/Auth block costs 54.3\% less under WTB.

\paragraph{BulkPR-Bench.}
WTB completes 6/10, Git/Auth 3/9, \cont{} 1/10, and \mas{} 2/10. WTB accepts all eight candidates entering E4. The corresponding rates are 4/6, 1/9, and 2/9. Thus, E4 contributes particularly large losses for the adapted research-code controls. WTB spends \$2.3017 in total, while Git/Auth spends less at \$1.6399 and completes fewer trajectories. Cost per success is \$0.3836 for WTB and \$0.5466 for Git/Auth. Across their three joint successes, WTB's cost reductions range from 1.2\% to 33.1\%.

\paragraph{Terminal-Bench.}
WTB completes 6/9, Git/Auth 0/10, \cont{} 2/10, and \mas{} 2/10. WTB accepts 8/9 entrants at E4, compared with 3/9, 2/5, and 2/9. Git/Auth accepts two of its three E5 entrants and loses both at E6. WTB's total cost is \$2.1263 and cost per success \$0.3544. \cont{} uses less total cost at \$1.5764 and costs \$0.7882 per success. The one jointly completed \mas{} block costs 41.9\% less under WTB. These comparisons preserve the distinction between assigned-workload expense and expense conditional on joint completion.

\subsection{Failure categories and implementation behavior}
E4 acceptance requires a material chain covering both contributing writers, exact execution-attempt binding, and renewed authority checks at publication and recovery. WTB connects native DSH receipts to its publication path. Comparator implementations sometimes achieve the same contract; their successful rows establish feasible adaptations.

\begin{table}[htbp]
\centering
\caption{E4 outcome decomposition. Bound: budget-bounded method failure; scope: invalid coding scope; envelope: invalid protocol envelope; blocked: predecessor failure. All five columns partition the observed trajectories in each row.}
\begin{tabular}{@{}llrrrrr@{}}
\toprule
Source & Method & Accepted & Bound & Scope & Envelope & Blocked \\
\midrule
\tableinput tables/e4_failures.tex
\bottomrule
\end{tabular}
\end{table}

Application code growth describes the successful implementations' integration footprint. It is conditioned on success and excludes native-library construction. WTB supplies much of the checked publication behavior through its library, while application adapters compose task payloads and material references. Small application diffs therefore measure reuse from the supplied starting point. They cannot establish a reduction in total system complexity or development effort.

\begin{table}[htbp]
\centering
\caption{Successful-trajectory code changes and actions: mean application lines added, test lines added, and actions. Rows without success have no conditional mean.}
\begin{tabular}{@{}llrrrr@{}}
\toprule
Source & Method & Successes & App added & Test added & Actions \\
\midrule
\tableinput tables/code_changes.tex
\bottomrule
\end{tabular}
\end{table}

\subsection{Paired uncertainty and joint-success cost}
\begin{table}[htbp]
\centering
\caption{All nine continuous paired comparisons. W-only and B-only count discordant successes; both counts joint successes. Exact two-sided McNemar tests are Holm-adjusted across all nine.}
\begin{tabular}{@{}llrrrrrr@{}}
\toprule
Source & Baseline & Pairs & W-only & B-only & Both & $p$ & Holm $p$ \\
\midrule
\tableinput tables/continuous_paired.tex
\bottomrule
\end{tabular}
\end{table}

Every source shows more complete trajectories for WTB in the observed cohort. No adjusted comparison passes 0.05; the smallest adjusted value is 0.28125. This combination supports a consistent descriptive direction across three fixed tasks while leaving statistical uncertainty substantial. Repeated blocks quantify model/coding variability for these tasks. Broader benchmark-distribution claims require more independently sampled tasks.

\begin{table}[htbp]
\centering
\caption{All six jointly completed task--block cost pairs. Conditioning on both methods' success restricts the interpretation to this subset.}
\begin{tabular}{@{}llrrrr@{}}
\toprule
Source & Baseline & Block & WTB (\$) & Baseline (\$) & Reduction \\
\midrule
\tableinput tables/common_costs.tex
\bottomrule
\end{tabular}
\end{table}

\subsection{Storage and qualification evidence}
All twelve starting implementations pass 117 qualification assertions each, totaling 1,404 assertions across E0--E2. DSH execution and restart work on both JSONL and SQLite. Six closed-store, same-backend location moves recover twelve exact receipts. All four methods pass 6/6 backend-contract checks on each of the three tasks, giving 72 method--task--check outcomes, with zero additional application changes. These checks concern a shared substrate capability. Cross-backend import and binary-version interoperability require separate testing.

\section{Benchmark source tasks and task-level mechanism checks}
\label{app:sources}

\subsection{Source provenance and functional outcomes}
The source records distinguish upstream benchmark revisions from frozen workload revisions. TeamBench's upstream revision is \code{d185aef}; its \code{CROSS2\_schema\_evolution} workload uses the adapted two-file base \code{f32f2c2}. BulkPR-Bench's upstream revision is \code{f4f2ea2}; the paired \code{attrs-PR-07-PR-25} workload uses attrs at \code{40487eb}. Terminal-Bench's upstream revision is \code{452bf30}; its \code{batched-\allowbreak eval-\allowbreak parity} workload uses the adapted two-file base \code{2f8035e}. The task mappings record pre-applied changes and parent bases. TeamBench material comes from an adapter-generated solution verified by the task grader; the other sources use reference solutions or patches. All three two-producer decompositions are experimental file-group adaptations. The upstream TeamBench benchmark defines Planner/Executor/Verifier roles, distinct from the experimental writer groups.

The source captures record 12/12 TeamBench functional checks, 164 passing local attrs tests, and five passing Terminal-Bench tests with reward one. BulkPR's saved local tests provide functional evidence for the constructed source workload; the official hidden-verifier result is unavailable. The governance studies freeze task material and reference changes, then ask whether downstream agents can implement publication, recovery, and maintenance under a common contract. These outcomes retain the upstream task context without becoming benchmark-wide issue-solving scores.

\subsection{Task-level W1 preliminary replays}
The three fixed source tasks also have deterministic task-level W1 mechanism replays. Without an additional guard, unsafe effects occur in 1/3 TeamBench cases, 2/3 BulkPR cases, and 1/3 Terminal cases, totaling 4/9. WTB and the fresh-authority reference each produce zero unsafe effects in the same nine cases. Legal-effect coverage is 2/2, 1/1, and 2/2 by task, totaling 5/5.

These small preliminary replays provide a direct source-task illustration of the mechanism and preserve their exploratory status. They enter neither the 128-cell authority matrix nor the 960-case runtime ablation denominator. The formal runtime studies share a cross-runtime contract across sources; each source contributes separate task-level illustration and downstream coding evidence.

\section{Reconstruction and reproducibility scope}
\label{app:protocol}

\paragraph{Numeric source map.}
The manuscript table/figure builder reads compact saved results. The following paths are relative to the companion experiment package. Its name map resolves historical identifiers to the descriptive study names used throughout this paper.
\begin{itemize}
\item Continuous adaptation: \path{submission/v23-continuous/analysis/publication-analysis.json.gz}.
\item Prospective SWE: \path{WTB/benchmark-experiment/derived/prospective/}, containing \path{analysis.json} and \path{trajectories.json}.
\item Source adoption: each source's \path{source-derived-adoption-cost/} directory contains \path{stage-results.csv} and \path{historical-summary.csv}.
\item Original mechanism: \path{WTB/original-mechanism/derived/main-mechanism/report-data.json}; saved planning and effects are in that family's \path{raw/main-mechanism/} directory.
\item Equal-contract comparison: \path{WTB/supplementary-experiments/derived/layered/}.
\item Runtime studies: \path{WTB/runtime-mechanism-validation/}; live rows are in \path{raw/live-interleaving/results.json} within that directory.
\end{itemize}
Each family report explains its protocol, outcomes, and reconstruction entry points.

\paragraph{Offline manuscript build.}
The source archive includes all generated table fragments, the quantitative vector figure, and three editable TikZ diagrams. Running PDFLaTeX, BibTeX, and PDFLaTeX twice rebuilds the paper without model calls or access to the experiment archives; the repository build script performs these steps. With the compact experiment package available at the documented location, \code{python build\_assets.py} regenerates table fragments and the trajectory figure from saved numeric records. It performs no manuscript prose generation and no paid experiment execution.

\paragraph{Re-execution boundary.}
The compact package supports offline reconstruction of recorded results and supplies selected integration and capture code. A fresh end-to-end execution additionally requires the WTB SDK, native runtime and comparator libraries, experiment controllers, frozen application starts, and container environments; the compact analysis distribution does not bundle that complete execution stack. The experimental comparisons concern the recorded model versions and supplied mechanisms. Changing provider behavior, native-library versions, prompts, or acceptance contracts defines a new experimental condition.

\end{document}